\documentclass{article}
\usepackage{ijcai19}

\usepackage[perpage]{footmisc}

\usepackage{array, booktabs}
\usepackage{tikz}
\usepackage{mathrsfs} % command for more caligraphic \mathscr{L}, togather with this package
\usepackage{amsmath}
\usepackage{amssymb}
\usepackage{helvet}  %Required
\usepackage{courier}  %Required
\usepackage{array}  % Needed for the "comment" environment to make LaTeX comments
\usepackage{xspace}
\usepackage{cite}
\usepackage{amsthm}
\usepackage{algorithmic}
\usepackage[ruled,linesnumbered,vlined]{algorithm2e}
\usepackage{caption}
\usepackage{subcaption}
\usepackage{enumerate}% http://ctan.org/pkg/enumerate

\usepackage{times}
\usepackage{soul}
\usepackage{url}
\usepackage[hidelinks]{hyperref}
\usepackage[utf8]{inputenc}
\newcommand{\VV}{\mathcal{V}}
\newcommand{\A}{\mathcal{A}}

\newtheorem{theorem}{Theorem}
\newtheorem{lemma}{Lemma}
\newtheorem{observation}{Observation}
\newtheorem{example}{Example}
\newtheorem{defnt}{Definition}
\usepackage{todonotes}

\title{Incentivising Participation in Liquid Democracy with Breadth-First Delegation\footnote{This work was supported by the EPSRC grant EP/P031811/1.}}

\author{
Grammateia Kotsialou$^1$
\and
Luke Riley$^2$
\affiliations
$^1$Department of Political Economy, King's College London, UK \\
$^2$Department of Informatics, King's College London, UK
\emails
\{grammateia.kotsialou, luke.riley\}@kcl.ac.uk.
}
\begin{document}

\maketitle

\begin{abstract}
 Liquid democracy allows members of an electorate to either  directly vote over 
 %the available election 
 alternatives, or  delegate their voting rights to someone they trust. 
  Most of the liquid democracy literature and implementations allow each voter to nominate only one delegate per election.  However, if that delegate abstains, the voting rights assigned to her are left unused. To minimise the number of unused delegations, it has been suggested that each voter should declare a personal ranking  over  voters she trusts.    In this paper, we show that even if personal rankings over voters are declared,  the standard delegation method of  liquid democracy
  remains problematic. 
  More specifically, we show that when personal rankings over voters are declared,  
   it could be undesirable to receive delegated voting rights, which is contrary to what  liquid democracy fundamentally relies on.
 To solve this issue, we propose a new method to delegate voting rights in an election, called \emph{breadth-first delegation}. 
  Additionally, the proposed method prioritises assigning voting rights to individuals closely connected to the voters who delegate.
\end{abstract}

\vspace{-4mm}

\section{Introduction} \label{sec:intro}

Liquid democracy is a middle ground between direct and representative democracy, as it allows each member of the electorate to directly vote on a topic, or temporarily choose a representative by delegating her voting rights to another voter. Therefore individuals who are either apathetic for an election, or trust the knowledge of another voter more than their own, can still have an impact on the election result  (through delegating). An individual who casts a vote for themselves and for others is known as a \emph{guru} (\citeauthor{ChristoffG17} \citeyear{ChristoffG17}).
 Liquid democracy  has recently started gaining attention in a few domains which we discuss to show  an overview of the general societal interest. In the political domain, local parties such as the Pirate Party in Germany,  Demoex in Sweden, the Net Party in Argentina and Partido de Internet in Spain  have been experimenting with liquid democracy implementations. Additionally, the  local governments of the  London Southwark borough and the Italian cities Turino and San Dona di Piave are working on integrating liquid democracy for community engagement processes (\citeauthor{BoellaFGKNNSSST18} \citeyear{BoellaFGKNNSSST18}). In the technology domain, the online platform LiquidFeedback uses a  liquid democracy system where a user selects a single guru for different topics (\citeauthor{LFBook} \citeyear{LFBook}; \citeauthor{KlingKHSS15} \citeyear{KlingKHSS15}). 
 Another prominent online example  is GoogleVotes (\citeauthor{GoogleVotes} \citeyear{GoogleVotes}),
 %which was run as  internal trial  mostly centered around voting over different catering options. which allowed delegation advertisements to be made through one of their corporate networks, 
 where each user  wishing to delegate can select a ranking over other voters.

 Regardless of the increasing interest in liquid democracy, there exists outstanding theoretical issues. This work focuses on liquid democracy systems where each individual  wishing to delegate can select a ranking over other voters.  In such systems, given the common  interpretation that delegations of voting rights are multi-step and transitive\footnote{If a voter $i$ delegates to a voter $j$, then $i$ transfers to $j$ the voting rights of herself and all the others that had been delegated to $i$.}, we observe that: searching for a guru  follows a depth-first search in a graph that illustrates all delegation preferences within an electorate, e.g. nodes represent the voters and directed edges the delegation choices for each voter.  For this reason, we name this standard approach of delegating voting rights as  \emph{depth-first delegation}. Despite its common acceptance, we came across an important disadvantage for 
 this rule. We show that when depth-first delegation is used, it could be undesirable to receive the voting rights of someone else. At this point, we emphasize that  disincentivising voters to participate as gurus is in contrast to the ideology of liquid democracy. Motivated by this, we propose a new rule for delegating voting rights, the \emph{breadth-first delegation}, which guarantees that casting voters (those who do not delegate or abstain) weakly prefer to  receive delegated voting rights, i.e. to participate as a guru.

We outline this paper as follows: Section \ref{sec:intro} discusses the latest applied and theoretical developments in liquid democracy and gives our model's preliminaries. In Sections \ref{sec:Delegation} \& \ref{sec:Participation}, we define delegation graphs, delegation rules and two types of participation. Section \ref{sec:delrulepref}  formally introduces a new delegation rule while Section \ref{sec:delrulesegs} compares this rule with the standard one. Finally,  Section \ref{sec:con} concludes this work.

\subsection{Related work} \label{sec:RelatedWork}

There currently exists a lack of theoretical analysis on liquid democracy. However, we summarise the main differences of our work to the main undertaking so far. 

As outlined by~\citeauthor{Brill18} (\citeyear{Brill18}), one of the main ongoing issues in liquid democracy is how to handle personal rankings over voters. His work discusses possible solutions around this issue without giving a formalised model, which this paper does. For two election alternatives and a ground truth on which  the correct one is,~\citeauthor{KahngMP18} (\citeyear{KahngMP18}) find that: (a) there is no decentralised liquid democracy delegation rule that is guaranteed to outperform direct democracy and (b) there is a centralised liquid democracy delegation rule that is guaranteed to outperform direct democracy as long as voters are not completely misinformed or perfectly informed about the ground truth.
In comparison, our model can be used in a wider variety of elections, as it allows for multiple alternatives and no ground truth. Additionally our delegation rules can be used in a central or decentralised manner, thus the negative result (a) does not apply to our paper. The work  of \citeauthor{ChristoffG17} (\citeyear{ChristoffG17})  focuses on the existence
of delegation cycles and inconsistencies that can occur when there are several binary issues to be voted on with a  different guru assigned for each issue.  In comparison, we avoid delegation cycles by stating that a delegation chain (a path from a delegating voter to their assigned guru) cannot include the same voter more than once. Furthermore, individual rationality issues between multiple elections is out of scope for this work.  Last, \citeauthor{BrillT18} (\citeyear{BrillT18}) introduce a special case of \citeauthor{ChristoffG17}'s model, which allows a single voter to be assigned several gurus. However, our model assigns one guru per voter.

Similar to our work, GoogleVotes  (\citeauthor{GoogleVotes} \citeyear{GoogleVotes}) allows a user to select a ranking over other voters and uses, what the authors describe as,  a back-track breadth first search  to assign a guru to a voter. We cannot complete a more comprehensive comparison to GoogleVotes as they have published only a general description of their system (without a formal model). However, we know that their delegation rule is different  to our proposed breadth-first delegation rule as in the Tally/Coverage section of their video example (from minute 32 of \citeauthor{GoogleVotesPresentation} \citeyear{GoogleVotesPresentation}), their rule assigns guru $C$ to voter $F$, while our rule would assign guru $A$ to voter $F$.

\subsection{Preliminaries}
\label{sec:Model}

Consider a set of voters $\VV$ and a set of alternatives or outcomes $\A$.  
The set of possible electorates is given by ${\cal E}(\VV) = 2^\VV \backslash \{ \emptyset \}$, i.e. non-empty subsets of $\VV$. In our model, for every election there are three sets of electorates $V^a, V^c, V^d \in {\cal E}(\VV)$ such that $V^a \cap V^c \cap V^d = \emptyset$ and $V^a \cup V^c \cup V^d = \VV$, 
where sets $V^a$, $V^c$, $V^d$ consist of those who abstain, cast a vote and  wish to delegate their voting rights, respectively. %to members of $V^c$. 

A \emph{preference relation over alternatives} for a voter $i \in \VV$ is denoted by $\succ^{\A}_i$ and is a binary relation on $\A$, i.e.: for $x,y\in A$ with $x \neq y$, the expression $x \succ^{\A}_i y$  indicates that voter $i$ strictly prefers alternative $x$ over alternative $y$. A   \emph{preference relation over voters}  for voter $i \in \VV$ is denoted by $\succ^{\VV}_i$ and is a binary relation on $\VV$, i.e.: for $i, x,y \in \VV$ with\footnote{A voter cannot include herself in her preference relation over voters.} $i \neq x, y$ and $x \neq y$, the expression $x \succ^{\VV}_i y$  indicates that voter $i$ strictly prefers to delegate her voting rights to voter $x$ instead of voter $y$. For both preference relations, we allow an index to identify ranking positions e.g. for any $i \in \VV^d$, her $m$-th preferred voter is denoted by $\succ^\VV_{i, m}$.

Consider a binary relation on a set $W$ given by $\succ^{W}_i$. Then  $\succ^{W}_i$ is:
%\begin{itemize}
%\item 
$(a)$ \emph{complete} iff for every pair $x, y \in W$ either $x \succ^{W}_i y$ or $y \succ^{W}_i x$ holds,
%\item 
$(b)$ \emph{antisymmetric} iff for every pair $x, y \in W$,  if $x \succ^{W}_i y$ then $y \succ^{W}_i x$ does not hold, and
%\item 
$(c)$ \emph{transitive} iff for all $x, y, z \in W$,  if  $x \succ^{W}_i y$ and $y \succ^{W}_i z$,  then $x \succ^{W}_i z$.
%\end{itemize}
Both preference relations over alternatives and  preference relations over voters are  antisymmetric and transitive but not complete (we do not enforce voters to rank every other member of the electorate as we consider this an unrealistic  scenario for large electorates).

%We consider as unrealistic the scenario where each voter must rank every other member of the electorate, thus we do not require the latter relation to be complete.

%In a slight abuse of notation, we allow the rank of the delegates of $i$ to be accessed by indexing, so the most preferred delegate of $i$ will be $\succ^\VV_{i, 1}$, while the least preferred delegate of $i$ will be $\succ^\VV_{i, |\succ^\VV_i|}$. 
The set of all possible preference relations $\succ^{\A}_i$ and  $\succ^{\VV}_i$, for any $i \in \VV$,  are denoted by ${R^\A}$ and ${R^\VV}$, respectively. A \emph{preference profile over alternatives} is a function $P^\A:\cal E (\VV) \rightarrow $~2$^{R^\A}$, where $P^A(N)$  returns a set of preference relations over alternatives (maximum one for each voter in $N$). For example, given an electorate $N=\{i,j,k\}$, a preference profile  $P^A(N)$ could return $\{(i, \succ^{\A}_i),(j, \succ^\A_j)\}$, meaning that  agent $i$ and $j$ have been assigned a preference relation over alternatives  but $k$ has not (as $k$ is either delegating or abstaining).  Similarly, we define as a \emph{preference profile over voters} a function $P^\VV:\cal E (\VV) \rightarrow $~2$^{R^\VV}$, where $P^\VV(N)$  returns a set of preference relations over voters (maximum one for each voter in $N$). Given profiles $P^\A$ and  $P^\VV$, voters are assigned to the $V^a$, $V^c$ and $V^d$ electorates as follows.
 If $(i, \succ^\A_i) \in P^\A(N)$, we infer that  voter $i$ casts a vote according to $\succ^\A_i$ and therefore becomes a member of the casting electorate $V^c$.
If $(i, \succ^\A_i) \notin P^\A(N)$ and  $(i, \succ^\VV_i) \in P^\VV(N)$, then $i$ becomes a member of the delegating electorate $V^d$. 
If $(i, \succ^\A_i) \notin P^\A(N)$ and  $(i, \succ^\VV_i) \notin P^\VV(N)$, $i$ becomes a member of the abstaining electorate $V^a$.

Given an electorate $N$, adding or removing a preference relation over alternatives (or over voters) from a preference profile over alternatives $P^\A(N)$ (or over voters $P^\VV(N)$), is denoted as follows.
For a tuple $(i, \succ^\A_i) \in P^\A(N)$, a voter $j \in \VV~\backslash  ~N$ and $j$'s assigned preference relation over alternatives $\succ^\A_j\in R^\A$:
\begin{align}
& P^\A_{-i}(N) := P^\A(N)~\backslash ~\{ (i, \succ^\A_i)\}, \nonumber\\
& P^\A_{+ (j, \succ^\A_j)} (N):= P^\A(N)~ \cup ~\{(j, \succ^\A_j)\}.\nonumber
\end{align}
Similarly, for a tuple  $(i, \succ^\VV_i)$ $\in P^\VV(N)$, a voter $j \in \VV~ \backslash ~N$ and $j$'s assigned preference relation over voters $\succ^\VV_j \in R^\VV$:
\begin{align}
&P^\VV_{-i}(N) := P^\VV(N)~\backslash ~\{ (i, \succ^\VV_i)\}, \nonumber \\
&P^\VV_{+(j, \succ^\VV_j)}(N) := P^\VV(N)~ \cup ~\{(j, \succ^\VV_j)\}. \nonumber
\end{align}
To simplify the above, we will be using the notation $P^\A, P^\VV, P^\A_{-i}$, $P^\A_{+ j}$, $P^\VV_{-i}$ and $P^\VV_{+j}$, accordingly.

% Given an electorate $N \in \cal E(\VV)$, the set of all preference profiles over alternatives for $N$ is given by $\cal P^{\A, N}$, while the set of all preference profiles over voters for $N$ is given by $\cal P^{\VV, N}$.

% Note that the vote of an agent $j \in V^d$  may not be used in an election, as an appropriate guru may not be found for $j$ according to the delegation rule function being used, formally described later in Definition~\ref{ref:DR}.

\section{Delegation graph and delegation rules} \label{sec:Delegation}

We use a graph to model possible delegations between voters:
\begin{defnt}A \emph{delegation graph} is  a weighted directed graph $G=(\VV, E, w)$ where:

\begin{itemize}
\itemsep0cm
\item $\VV$ is  the set of nodes representing the agents registered as voters;
\item $E$ is the set of directed edges representing delegations  between voters; and 
\item $w$ is the weight function $w : E \mapsto \mathbb{N}$ that assigns a value to an edge $(i, j)$ equal to $i$'s preference ranking of $j$.
\end{itemize}
\end{defnt}

To generate a delegation graph, we use the following:

\begin{defnt} Define as $g$ the \emph{delegation graph function} which takes as input a preference profile over voters $P^\VV$ and returns the related delegation graph $G=(\VV, E, w)$ with the following property: for every $ i, j \in \VV$ and  $i \neq j$,
\begin{itemize}
\item if there exists $(i, \succ^\VV_{i}) \in P^\VV$ with $\succ^\VV_{i,x} = j$; then there exists $(i, j) \in E$ where $w((i, j)) = x$.
\end{itemize}

\end{defnt}

We can evaluate a delegation graph through the following:

\begin{defnt} A  \emph{delegation rule function} $d$ takes as input a preference profile over alternatives $P^\A$ together with a delegation graph $G$, and returns another preference profile over alternatives $\hat{P}^{\A}$ that resolves delegations. More specifically,  
\begin{itemize}
\itemsep0cm
\item if $(i, \succ^\A_i) \in \hat{P}^{\A}$, then $i$  casts her vote,
\item if $(i, \succ^\A_j) \in \hat{P}^{\A}$ for a  voter $j \neq i$, then $j$ becomes $i$'s final delegate, i.e. her guru,
\item  if $(i, \succ^\A_k) \notin \hat{P}^{\A}$ for any $k \in \VV$, then $i$ abstains. 
\end{itemize}
\label{ref:DR}
\end{defnt}
\noindent For each voter $i \in \VV$, a delegation rule analyses  the subtree of the delegation graph rooted at node $i$ and decides whether $i$ casts, delegates or abstains. If voter $i$ is found to delegate, the chosen delegation rule function will traverse $i$'s subtree %(according to its rules) 
to find $i$'s guru. 

To get  the outcome of an election, we use a  voting rule function $f$. In our model, $f$ takes as input the modified preference profile over alternatives $\hat{P}^{\A}$ (which incorporates delegations) and  returns a single winner or a ranking over  alternatives (depending on the voting rule used).  In Section \ref{sec:delrulesegs}, we show that the output of the voting rule depends on the chosen delegation rule, meaning that we could get different election results when only the delegation rule function is different, i.e. $f(d(P^\A, g(P^\VV))) \neq f(d'(P^\A, g(P^\VV)))$, when  $d \neq d'$.

%\begin{defnt} A \emph{Liquid Democracy Tuple} is a defined as $LD = \langle \VV, \A, R^\A,  R^\VV,  D, F\rangle$, which contains a set of possible voters $\VV$, a set of possible alternatives $\A$, a preference profile over alternatives $R^A$, a preference profile over gurus $R^\VV$,  a delegation rule $D$ and a voting rule $F$. \end{defnt}

\section{Cast and guru participation} \label{sec:Participation}
The key property that we  investigate is participation. The participation property holds if a voter, by joining an electorate, is at least as satisfied as before joining. This property has been defined only in the context of vote casting (\citeauthor{Fishburn} \citeyear{Fishburn}; \citeauthor{Moulin} \citeyear{Moulin}). Due to the addition of delegations in our model, we  establish two separate definitions of participation reflecting this new functionality\footnote{There could be other interesting participation properties for liquid democracy, such as incentivising deviation from delegating to casting. But this is out of the paper's scope, as we focus on finding delegation rules that weakly benefit casting voters who become gurus.}.

%(inspired by the participation definition of~\citeauthor{BrandtGP16} \citeyear{BrandtGP16}). 

For both of the following definitions, note that for an electorate $N \in {\cal E}(\VV)$, the set of all preference profiles over alternatives is given by $\cal{P}$\hspace{-0mm}$^{\A,N}$, while
the set of all preference profiles over delegates is given by
$\cal{P}$\hspace{-0mm}$^{\VV,N}$.

A voting rule $f$ satisfies the \emph{cast participation} property when every voter $i \in \VV$ weakly prefers joining any possible voting electorate $V^c$  compared to abstaining and regardless of who is in the delegating electorate $V^d$. 

%Therefore by voting,  agent $i$ will always either: have no affect on the election result; or will change the result for $i$'s benefit. 

\begin{defnt}
The \emph{Cast Participation} property holds for a voting rule $f$ iff:
\begin{align*}
f(d(P^\A, g(P^\VV))) \succeq^{\A}_i f(d(P^\A_{-i}, g(P^\VV))),
\end{align*}  
for every possible disjoint casting and delegating electorates $V^c, V^d \in {\cal E}(\VV)$, where $i \in V^c$, and every possible preference profile for these electorates $P^A \in {\cal P}^{\A, V^c}$ and $P^\VV \in {\cal P}^{\VV, V^d}$. 
\end{defnt} 

For any casting and delegating electorates $V^c$ and $V^d$, a voting rule $f$ satisfies the \emph{guru participation} property when any voter $i \in V^c$ weakly benefits from receiving  additional voting rights of any voter $j \in \VV$. 

%A voting rule $f$ satisfies the \emph{delegation participation} property when every voter $j \in \VV$ weakly prefers joining any possible delegating electorate $V^d$  compared to abstaining and regardless of who is in the casting electorate $V^c$. Voter $j$ weakly prefers to delegate as this participation property guarantees that $j$'s assigned guru $i \in \VV$ weakly benefits from the delegation. 

%The next definition, Guru Participation, holds for a voting rule $f$ if each vote casting agent (potential guru) $i \in \VV$ weakly prefers receiving a delegated vote from agent $j \in \VV$ compared to not receiving the delegated vote, given any possible voting electorate $V^c$ and any possible non conflicting delegating electorate $V^d$. Therefore when $i$ receives $j$'s delegated vote, the election result will either:  not change; or will change to $i$'s benefit. 

\begin{defnt} 
The \emph{Guru Participation} property holds for a voting rule $f$ iff:
\begin{align} \label{guru_p}
f(d(P^\A, g(P^\VV))) \succeq^{\A}_i f(d(P^\A, g(P^\VV_{-j}))),
\end{align}  
for every possible disjoint casting and delegating electorates $V^c, V^d \in {\cal E}(\VV)$, where $i \in V^c$ , $j \in V^d$, and every possible profile $P^A \in {\cal P}^{\A, V^c}$ and $P^\VV \in {\cal P}^{\VV, V^d}$ that assign $j$'s vote to guru $i$, i.e. $(j, \succ^\A_i) \in d(P^\A, g(P^\VV))$.
\end{defnt}   

%\begin{defnt}
%The \emph{Guru Participation} property holds when a voter $i$, who was in the abstaining electorate $V^a$, joins the delegating electorate $V^d$, delegates her vote to $j$ and $j$ is no worst off according to  $P_{j}^\A$. %Formally:\\
% $F(D(\LL^\A-i
%,G(\LL^\VV-i))) \preceq_j F(D(\LL^\A-i,G(\LL^\VV)))$.	
%\end{defnt}

Let $F$ be the set of all voting rules. It is known  that only a subset $\bar{F}\subset F$ satisfy (cast) participation. For example,   \citeauthor{Fishburn} (\citeyear{Fishburn}) show that  single transferable vote  does not satisfy (cast) participation, while \citeauthor{Moulin} (\citeyear{Moulin}) shows there is no Condorcet-consistent voting  rule satisfying this property  given 25 or more voters. We explores guru participation for voting rules in $\bar{F}$, and
our results build on the following observation, which we intuitively descibe: if a new voter $j$ joins the delegating electorate and only one voter $i$  from the casting electorate increases the number of times she is assigned as a guru, then guru participation is satisfied.

\begin{observation}
%Consider 
%a tuple $LD^x$ that contains a voting rule satisfying cast participation (i.e. $F \in \bar{\mathcal{F}}$) and this tuple gives a preference profile over alternatives, that incorporates delegations, of
 Consider $i, j \in \VV$, profile $\hat{P}^\A$ returned by $d(P^\A, g(P^\VV))$ and profile $\hat{P}^{'\A}$ returned by  $d(P^\A, g(P^\VV_{+j}))$,
%Now consider another tuple $LD^y$ that is identical to $LD^x$ apart from: agent $j$ has joined the delegating electorate from the abstaining electorate, giving $\bar{R^{\A'}} = D(R^\A,G(R^\VV+(j, \succ^\VV_j)))$;
 where $i$ has been assigned as $j$'s guru, i.e. $(j, \succ^\VV_i) \in \hat{P}^{'\A}$. 
 %\vspace{mm}
 Guru $i$ becomes weakly better off after $j$ delegates if the following holds. For every $k \in \VV$: 
\begin{enumerate}[$(a)$]
\itemsep0cm
\item $k$'s vote is assigned to  guru $l \in \VV$ by both returned preference profiles, i.e. $(k, \succ^\VV_l) \in \hat{P}^{\A} \cap \hat{P}^{'\A}$, or
\item $k$'s vote is assigned to guru $i$ after $j$ joins the delegating electorate, i.e. $(k, \succ^\A_i) \in \hat{P}^{'\A}$,
\end{enumerate}
\end{observation}

\section{Introducing breadth-first delegation}
\label{sec:delrulepref}

%Delegation rules use a delegation graph $G$ provides a way to visualise the delegation connections in an electorate according to the preference profile over voters $P^\VV$. 

Recall that liquid democracy allows for multi-step delegations. Therefore, the guru of any $i \in V^d$ could be any voter $j \in V^c$ who is  in the sub tree of the delegation graph with root $i$. In addition, the assigned  guru $j$ may not be included in $i$'s preference relation $\succ^\VV_i$, i.e. it could be  that $\nexists ~x$ such that $\succ_{i,x}^\VV = j$. In this case, there is at least one intermediate delegator between voter $i$ and the assigned guru $j$. To find the exact intermediate delegators, we introduce delegation chains. A \emph{delegation chain} for a voter $i \in V^d$ starts with $i$, then lists  the intermediate voters in $V^d$ who have further delegated $i$'s voting rights and 
ends with $i$'s assigned guru $j \in V^c$. These chains (see Definition~\ref{deft:dc}) must satisfy the following conditions:  $(a)$ no voter occurs more than once in the chain (to avoid infinite delegation cycles that could otherwise occur) and  $(b)$  each member of the chain must be linked to the next member  through an edge in the delegation graph, which is generated from the given  preference profile over voters.

\begin{defnt} Given profiles $P^\A$ and $P^\VV$,  a voter $i \in V^d$ and her  guru $j \in V^c$, we define 
 a \emph{Delegation Chain} for $i$ to be   an ordered tuple   $C_i = \langle i, \ldots, j \rangle$ with the following properties:  
\begin{enumerate}[$(a)$]
\itemsep0cm
%\item $i \in V^d$, $j \in V^c$;
\item  For  an integer $x \in [1, |C_i|]$, let $C_{i, x}$ indicate the voter at the $x$-th position in $C_i$. Then for  any pair of integers  $x,y \in[1, |C_i|]$ with $x \neq y$, 
$$C_{i,x} \neq C_{i,y}.$$
%\item $\nexists$ $k \in C_i$ where $k \in V^c$ and $k \neq C_{i, |C_i|}$;
\item  For each integer $z \in [1, |C_i|-1]$, there exists an edge
$$(C_{i, z},  C_{i, z+1}) \in E \in  g(P^\VV).$$

 %where $E$ is the set of edges of the delegation graph returned by $g(P^\VV)$ with a   preference profile over voters $P^\VV$.
\end{enumerate}
\label{deft:dc}
\end{defnt} 
Observe that the $x$ in the expression $C_{i, x}$  also indicates how deep the voter $C_{i, x}$ is in the delegation graph subtree rooted with 
$i$. Thus sometimes we refer to $x$ as the depth of $C_{i, x}$ in $C_i$. The function $w$ takes as input a delegation chain and returns a list of the weights assigned to each edge among voters in $C_i$, that is, $w(C_i) = [w(C_i)_1, \ldots, w(C_i)_x, \ldots, w(C_i)_{n-1}]$, where $w(C_i)_x$ is the weight of edge $(C_{i, x}, C_{i, x+1})$  and $n= |C_i|$.

 %If there is an attempt to return the weight at $w(C_{i})_f$ where $f \geq |C_{i}|$, then $w(C_{i})_f = |\VV|+1$.

Delegation chains  can be used as a tool to find  a guru for a voter $i \in V^d$. The standard interpretation of  liquid democracy delegations prioritises all possible delegation chains involving $i$ and $i$'s most preferred voter $\succ_{i,1}^\VV$ before all possible delegation chains involving $i$ and $i$'s second preferred voter $\succ_{i,2}^\VV$ and so on. Note that this priority rule hold for the deeper levels of the delegation graph subtree rooted at $i$. In other words, we observe that the standard way to select $i$' guru is to choose the first casting voter found through a depth first search in $i$'s subtree, which motivates the next definition.
%(\citeauthor{BrillT18} \citeyear{BrillT18}; \citeauthor{ChristoffG17} \citeyear{ChristoffG17}; \citeauthor{KahngMP18} \citeyear{KahngMP18})
%A complication arises when $i$'s most preferred voter $\succ_{i,1}^\VV$ also delegates. In this case what delegation chains should be considered? .....The standard concept in liquid democracy literature until now .....suggests that delegations of voting rights is transitive, meaning that, if $i \in V^d$ delegates to $j$ who also decides to delegate to someone $\in V^d$ and so on,....the voting power of $i$ could be transferred way further from her.

 A  \emph{depth-first delegation} rule $d^D$ assigns guru $j$ to  $i$ iff:  $(a)$ there is a delegation chain $C_i$ that can be formed from $i$ to $j$, and $(b)$ there is no other delegation chain $C_i'$ leading to a different guru $k$ that has a smaller weight at the earliest depth after the root, compared to $C_i$.

% A  \emph{depth-first delegation (DFD)} (see Definition~\ref{def:DFD}), assigns guru $j$ to  $i$ iff:  (1) there is a delegation chain $C_i$ that can be formed from $i$ to $j$; and (2) there is not another delegation chain $C_j'$ where (2,a) $C_i'$ leads to guru $k$, and (2,b) $C_i'$ uses a smaller weight delegation than $C_i$, before $C_i$ uses a smaller weight delegation than $C_i'$.

\begin{defnt}
\label{def:DFD}
For $i, j, k \in \VV$, a \emph{depth-first delegation rule} $d^{D}$ returns a profile $\hat{P}^{\A}$ with  $(i, \succ^\A_j) \in \hat{P}^{\A}$ iff $(a)$ and $(b)$ hold:  
\begin{enumerate}[$(a)$]
\itemsep0cm
%\item  $\exists~p \in V^d$, $\exists~ s \in V^c$;
\item $\exists~C_i$ with $C_{i, |C_i|} = j$,
\item  $\nexists~C_i'$ such that: 
\begin{enumerate}[$b1.$]
\item $C'_{i, |C_i'|} = k$ ~for  $k \neq j$, 
\item 
\begin{itemize}
\item $\exists y$: $w(C_i')_y < w(C_i)_y$  and  
\item $w(C_i')_x \leq w(C_i)_x$ ~for all $~0 < x < y$.
\end{itemize}
\end{enumerate}
\end{enumerate}
\end{defnt}

\begin{figure}
\centering
\begin{minipage}{.45\textwidth}
\begin{tabular}{cc}
\vspace{-.2cm}
\hspace{-.55cm}\includegraphics[width=.5\linewidth]{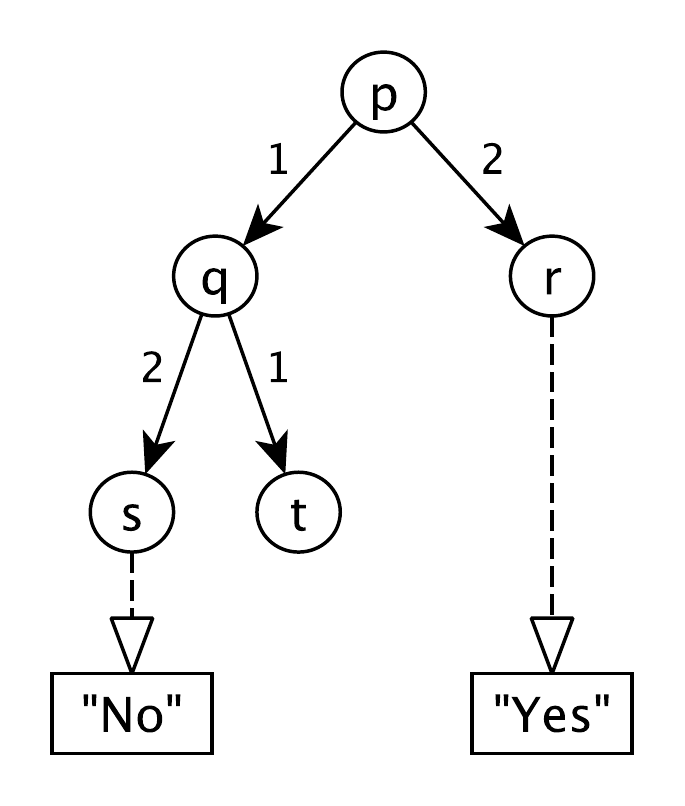}&\includegraphics[width=.5\linewidth]{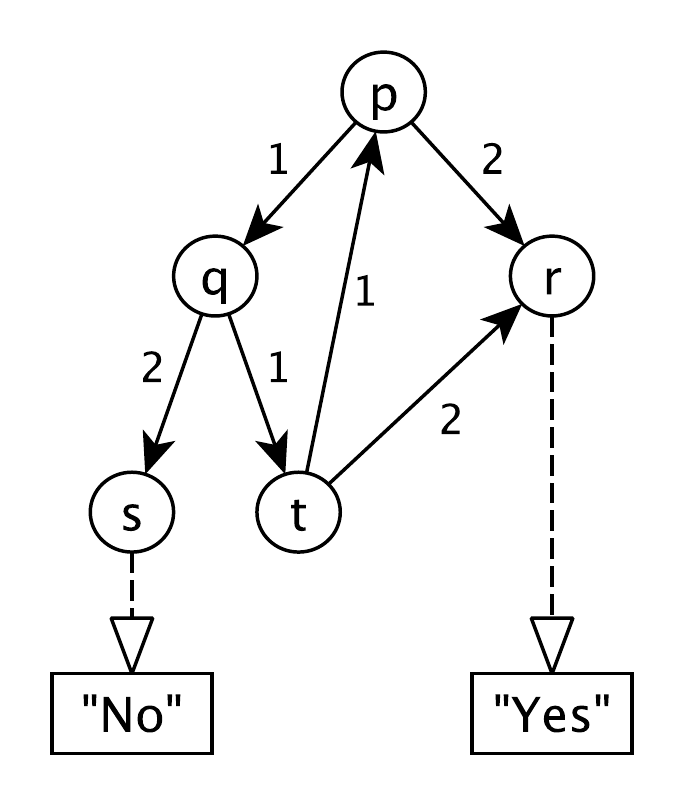}\\
\hspace{-.35cm}$(a)$&$(b)$
\end{tabular}
\captionof{figure}{$(a)$ A delegation graph with electorates $V^a=\{t\}$, $V^c = \{s, r\}$ and $V^d = \{p, q\}$, meaning that $t$ abstains, $s, r$ cast, and $p, q$ delegate. The preference relations over alternatives are: ``No" $\succ^\A_s$ ``Yes" and  ``Yes" $\succ^\A_r$ ``No".\quad 
  $(b)$ A delegation graph with electorates $V^a=\{\}$, $V^c = \{s, r\}$ and $V^d = \{p, q, t\}$. 
 The preference relations over alternatives remain the same. The only difference to $(a)$ is  that voter $t$ decides to delegate  with a preference relation over voters $p \succeq_t^\VV r$.}
  \label{fig:loop}
\end{minipage}
\end{figure}

\begin{example}
Consider the delegation graph  in Figure~\ref{fig:loop}~$(a)$. There are two delegation chains\footnote{Recall that $\langle p, q, t\rangle$ is not a valid delegation chain as $t \notin V^c$.} available for voter $p \in \VV$: $C_p = \langle p, r\rangle$ and $C_p' = \langle p, q, s\rangle$ with weights $w(C_p) = [2]$ and $w(C_p') = [1, 2]$, respectively.  The $d^D$ rule returns profile $\hat{P}^{\A}$ that assigns $s$ as the guru of $p$, i.e. $(p, \succ^\A_s) \in \hat{P}^{\A}$, due to inequality $w(C_p')_1 < w(C_p)_1$. Note that  $C_p'$ satisfies  Definition~\ref{def:DFD}.
\label{ex:dd}
\end{example}

In Example \ref{ex:dd}, $p$'s voting right is assigned to guru $s$, but why should $s$ (who is the second preference of $q$) outrank agent $p$'s explicit second preference $r$? This question gains even more importance the longer the  depth first delegation chain is.
Given this issue, we define a novel delegation rule that prioritises a voter's explicit preferences as follows. A  \emph{breadth-first delegation} rule $d^B$ assigns guru $j$ to  $i$ iff:  $(a)$ there is a delegation chain $C_i$ that can be formed from $i$ to $j$; and $(b)$ there is no other delegation chain $C_i'$ leading to a different guru $k$ with: either a shorter length or,  an equal length and a smaller weight at the earliest depth after the root, compared to $C_i$.

%which says that $i$'s vote should be assigned to guru $j$ iff: (i) $p$ delegates her vote and $s$ casts her vote; (ii) there is a delegation chain $DC(p)^x$ that can be formed from $p$ to $s$; (iii) there is not another voter $r$ and delegation chain $DC(p)^y$ where (iii, a) $DC(p)^y$ delegates $p$'s vote to $r$, and either (iii,b1) $DC(p)^y$ is shorter than $DC(p)^x$  or (iii,b2) $DC(p)^y$ is equal in size to $DC(p)^x$ but $DC(p)^y$ uses a smaller weight delegation than $DC(p)^x$ before $DC(p)^x$ uses a smaller weight delegation than $DC(p)^y$.

\begin{defnt}
\label{def:BFD}
 For $i,j,k \in \VV$, a \emph{breadth-first delegation rule} $d^{B}$ returns a profile $\hat{P}^{\A}$ with  $(i, \succ^\A_j) \in \hat{P}^{\A}$ iff $(a)$ and $(b)$ hold: 
\begin{enumerate}[$(a)$]
%\item $\exists~p \in V^d$, $\exists~s \in V^c$;
\item $\exists~C_i$ with $C_{i, |C_i|} = j$,
\item $\nexists~C_i'$ such that $C'_{i, |C_i|} = k$, for $k \neq j$, and
\begin{enumerate}[$b1.$]
\itemsep0cm
\item $|C_i'| < |C_i|$, or
\item  
\begin{itemize}
\item $|C_i'| = |C_i|$ and
\item $\exists y$: $w(C_i')_y < w(C_i)_y$  and  
\item $w(C_i')_x \leq w(C_i)_x$ ~for all $~0 < x < y$. 
\end{itemize}
\end{enumerate}
%\end{enumerate}

\end{enumerate}
   \label{deft:BFD}
\end{defnt}

\begin{example}
Consider the delegation graph  in Figure~\ref{fig:loop}~$(a)$. There are two delegation chains available for voter $p \in \VV$: $C_p = \langle p, r\rangle$ and $C_p' = \langle p, q, s\rangle$ with weights $w(C_p) = [2]$ and $w(C_p') = [1, 2]$, respectively.  The $d^B$ rule returns profile $\hat{P}^{\A}$ that assigns $r$ as the guru of $p$, i.e. $(p, \succ^\A_r) \in \hat{P}^{\A}$, due to inequality $|C_p| < |C_p'|$. Note that  $C_p$ satisfies  Definition~\ref{deft:BFD}.
\label{eg:3}
\end{example}

\section{Depth-first versus breadth-first delegation}
\label{sec:delrulesegs}

Through the next two examples, we show that different delegation rules can have different properties. More specifically, we present an instance where the depth-first delegation rule cannot guarantee guru participation, while the breadth-first delegation rule does.

 %Note that the numeric election results for both examples are displayed in Table~\ref{tab:electionResults}.

%The first example shows that by only changing the delegation rule function, the result of a liquid democracy game can change:

\begin{table} 
\centering
\renewcommand{\arraystretch}{1.2}
\begin{tabular}{ cccc } 
 \hline
\hspace{-1mm} \textsc{delegation graph} \hspace{0mm} & \hspace{0mm} \textsc{delegation rule} \hspace{-1mm} & \hspace{-1mm} Yes \hspace{-1mm} & \hspace{-1mm} No \hspace{-1mm}
\\ \hline 
 Figure~\ref{fig:loop} $(a)$ & $d^D$ & 1 & 3 \\ 
 Figure~\ref{fig:loop} $(b)$ & $d^D$ & 3 & 2 \\  
Figure~\ref{fig:loop} $(a)$ & $d^B$ & 2 & 2 \\  
Figure~\ref{fig:loop} $(b)$ & $d^B$ & 3  & 2 \\
\hline
 \end{tabular}
  \caption{Election results for Figure~\ref{fig:loop} when using either  the depth-first or the breadth-first delegation rule.} \label{tab:electionResults}

\end{table}

\begin{example}
Consider the delegation graph in Figure~\ref{fig:loop}(a) and all possible delegation chains available to each voter in $V^d$:  $C_p = \langle p, r\rangle$, $C_p' = \langle p, q, s\rangle$ and $C_q = \langle q, s \rangle$. Using rule $d^{D}$, voter $p$ is assigned guru $s$ through chain $C_p'$ (see Example~\ref{ex:dd}), while voter $q$ is also assigned guru $s$ through chain $C_q$. Therefore $d^{D}$ returns the preference profile over alternatives  $\{ (p, \succ^\A_s), (q, \succ^\A_s), (s, \succ^\A_s), (r, \succ^\A_r)\}$.  Using  rule $d^{B}$ instead,  voter $p$ is  assigned guru $r$ through $C_p$ (see Example~\ref{eg:3}), while $q$'s guru remains the same. Therefore $d^{B}$  returns another preference profile over alternatives: $\{(p, \succ^\A_r), (q, \succ^\A_s), (s, \succ^\A_s), (r, \succ^\A_r)\}$.

\label{eg:ExampleNoLoop}
\end{example}

In the next example we focus on the case where the previously abstaining voter $t$ decides to delegate and show that the election result is inversed only when  $d^D$ is used (see Table \ref{tab:electionResults}).

\begin{example} \label{eg:ExampleLoop}

Consider the delegation graph in Figure~\ref{fig:loop}(b) and all possible delegation chains available to each voter in $V^d$ with their respective edge weights:
\begin{center}
\renewcommand{\arraystretch}{1.3}
\begin{tabular}{ l l } 
 \textsc{delegation chain} & \textsc{edge weights}\\ 
%\midrule
$C_p= \langle p, r \rangle$ & $w(C_p)= [2]$, \\ 
%\midrule
$C_p'= \langle p, q, s \rangle$ & $w(C_p')= [1, 2]$,\\
%\midrule
$C_p''= \langle p, q, t, r \rangle$ & $w(C_p'')= [1, 1, 2]$, \\
%\midrule
$C_q= \langle q, s \rangle$ & $w(C_q)= [2]$, \\
%\midrule
$C_q'= \langle q, t, r \rangle$ & $w(C_q')= [1, 2]$, \\
%\midrule
$C_q''= \langle q, t, p, r \rangle$ & $w(C_q'')= [1, 1, 2]$, \\
%\midrule
$C_t= \langle t, r \rangle$ & $w(C_t)= [2]$, \\
%\midrule
$C_t'= \langle t, p, r \rangle$ & $w(C_t')= [1, 2]$, \\
%\midrule
$C_t''= \langle t, p, q, s \rangle$ & $w(C_t'')= [1, 1, 2]$. \\
\end{tabular}
 \end{center}

%$C_p = \langle p, r\rangle$ with $w(C_p) = [2]$, $C_p' = \langle p, q, s\rangle$ with $w(C_p') = [1, 2]$, $C_p'' = \langle p, q, t, e\rangle$ with $w(C_p'') = [1, 1, 2]$, $C_q = \langle q, s \rangle$ with $w(C_q') = [2]$, $C_q' = \langle q, t, r \rangle$ with $w(C_q') = [1, 2]$, $C_q'' = \langle q, t, p, r \rangle$ with $w(C_q'') = [1, 1, 2]$, $C_t = \langle t, r \rangle$ with $w(C_t) = [2]$, $C_t' = \langle t, p, r \rangle$ with $w(C_t') = [1, 2]$ and $C_t'' = \langle t, p, q, s \rangle$ with $w(C_t'') = [1, 1, 2]$.

Using rule $d^{D}$, observe that voter $p$ is assigned guru $r$ through the chain $C_p''$ ~due to ~$w(C_p'')_1 < w(C_p)_1$, $w(C_p'')_1 = w(C_p')_1$ ~and~ $w(C_p'')_2 < w(C_p')_2$. Voter $q$ is also assigned
 guru $r$ through chain $C_q''$~since~$w(C_q'')_1 < w(C_q)_1$, $w(C_q'')_1 = w(C_q')_1$ and $w(C_q'')_2 < w(C_q')_2$. Last, voter $t$ is assigned guru $s$ through chain $C_t''$ because~$w(C_t'')_1 < w(C_t)_1$, $w(C_t'')_1 = w(C_t')_1$~and~$w(C_t'')_2 < w(C_t')_2$.
 Therefore rule $d^{D}$ returns the preference profile over alternatives~$\{ (p, \succ^\A_r), (q, \succ^\A_r), (s, \succ^\A_s), (r, \succ^\A_r), (t, \succ^\A_s)\}$.
Using rule $d^{B}$ instead, voter $p$ is assigned guru $r$ through the chain $C_p$ due to inequalities $|C_p| < |C_p'| < |C_p''|$. Voter $q$ is assigned guru $s$ through $C_q$ due to  $|C_q| < |C_q'| < |C_q''|$ and voter $t$ is assigned guru $r$ through $C_t$ because of $|C_t| < |C_t'| < |C_t''|$. Therefore, rule $d^{B}$ returns
 the profile over alternatives~$\{(p, \succ^\A_r), (q, \succ^\A_s), (s, \succ^\A_s), (r, \succ^\A_r), (t, \succ^\A_r)\}$.
 \end{example}

%\section{Delegation Rule Properties} \label{sec:Proof}

\

Examples~\ref{eg:ExampleNoLoop} and~\ref{eg:ExampleLoop} show that guru participation may not hold for depth-first delegation when a cycle exists in the delegation graph. Due to this cycle, when $t$ joins the election, both  $r$ and $s$ receive new delegated voting rights, thus Observation 1 does not occur\footnote{Observation 1 states how guru participation can be satisfied when a voting rule satisfying cast participation is used: when a  voter joins the delegating electorate, if only one voter increase the number of times assigned as a guru, then this voter is weakly better off.}. We summarise the above for the set of voting rules satisfying cast participation  $\bar{F}$.

\begin{theorem}
%Given a depth-first delegation rule $d^{D}$, then guru participation is \emph{not guaranteed} to hold. 
Given a voting rule $f \in \bar{F}$,  guru participation is not guaranteed to hold when using the depth-first delegation rule $d^D$.
\label{theorem:1}
\end{theorem}

\begin{proof}
%By example. Take  

Consider the preference profile over alternatives and the preference profile over voters of Figure~\ref{fig:loop}$(b)$,
\begin{align*}
& P^\A = \{(r,\succ^\A_r), (s,\succ^\A_s)\} \\
& P^\VV = \{(p,\succ^\VV_p), (q,\succ^\VV_q), (t,\succ^\VV_t)\}, 
\end{align*}
where the preferences over alternatives for  $r$ and $s$ are: ``Yes'' $\succ^\A_r$ ``No", ``No"$\succ^\A_s$ ``Yes'' and the preferences over voters for $p,q,t$ are: $q \succ^\VV_p r$,~~$t \succ^\VV_q s$~and~$p \succ^\VV_t r$. We prove this theorem using Examples~\ref{eg:ExampleNoLoop} and ~\ref{eg:ExampleLoop}. 
In Example~\ref{eg:ExampleNoLoop}, where voter $t$ abstains, rule $d^D$ returns profile
\begin{align*}
d^{D}(P^\A,g(P^\VV_{-t}))= \{ (p, \succ^\A_s), (q, \succ^\A_s), (s, \succ^\A_s), (r, \succ^\A_r)\},
\end{align*}
which gives  three votes (via $s$) for alternative ``No" and one vote (via $r$) for alternative ``Yes" (see also Table \ref{tab:electionResults}). From Example~\ref{eg:ExampleLoop} where voter $t$ delegates, rule $d^D$ returns profile
\begin{align*}
& d^{D}(P^\A,g(P^\VV))= \\
& \{ (p, \succ^\A_r), (q, \succ^\A_r), (s, \succ^\A_s), (r, \succ^\A_r), (t, \succ^\A_s)\},
\end{align*}
which gives  three votes for ``Yes" and two votes for ``No".
Observe that the election result changes from ``No" to ``Yes" despite the fact that $t$ votes for ``No" through her guru $s$, i.e. $(t, \succ^\A_s) \in d^{D}(P^\A, g(P^\VV))$.  Note that due to the preference ``No"$\succ^\A_s$ ``Yes'', we get
\begin{align}\label{s}
f(d^{D}(P^\A,g(P^\VV))) \prec^{\A}_s f(d^{D}(P^\A,g(P^\VV_{-t}))),
\end{align}
where $f$ is a voting rule satisfying cast participation.
%Under the depth-first delegation rule $d^D$, when voter $t$ enters the election by delegating, the result changes from ``No" to ``Yes" despite the fact that $t$ votes for ``No" through her guru $s$, i.e. $(t, \succ^\A_s) \in d^{D}(P^\A, g(P^\VV))$) and $s$  votes for ``No". 
However, the preference expressed by \eqref{s} implies that guru $s$ becomes worst off after $t$ delegates to her, which violates the definition of guru participation  \eqref{guru_p}, proving this theorem.
%Theorem \ref{theorem:1}.
\end{proof}

We highlight that if a delegation graph has no cycle then guru participation is guaranteed to hold for  the depth-first delegation rule, which show through Lemma \ref{lem:DFDLemma1} and Theorem~\ref{theorem:1b}.

\begin{lemma}
When using depth-first delegation rule $d^D$,  if there is no cycle in the delegation graph then Observation 1 holds.
\label{lem:DFDLemma1}
\end{lemma}

\begin{proof}
Assume there exists a delegation graph with no cycles where Observation 1 does not hold. We show that the only case where Observation 1 does not hold is when a cycle exists. 

Recall that, by Observation 1,  guru participation is guaranteed to hold if whenever a voter $j$ joins the delegating electorate, there exists only one voter, say $i$, in the casting electorate who increases the number of times she becomes a guru. Consider another voter $k$ who changes her assigned guru  to a voter $l$ after $j$ joins the delegating electorate, where  $l \neq i$ and $k \neq j$. This means that, apart from $i$,  voter $l$ also  increases the times  she becomes a guru.
 Next we describe that, when $d^D$ is used, this case  can only arise through the following circumstance.
%(i) delegation chain $\langle k, ..., l \rangle$ existed before $j$ joined, in which case guru $l$  should have been chosen previously if $l$ really was the most preferred guru (and so Observation 1 would still hold in this case); 
Let guru $i$ be assigned to voter $j$ through  delegation chain  $C_j = \langle j, \ldots, i \rangle$ and  guru $l$ be  assigned to voter $k$ through  delegation chain $C_k= \langle k,...,j,...,l \rangle$.  Chain $C_k$ must pass through $j$ because all chains without $j$ are available  before $j$ delegates. Note that even if both chains pass through voter $j$, they end at different gurus. For $d^D$, this only occurs if there exists a voter $h$  with   $h \neq i, j, l$, such that
 \begin{align}
 &C_j= \langle j,\ldots, h, \ldots, i \rangle \quad \text{and} \label{cycle1} \\
 &C_k=\langle k,\ldots, h, \ldots, j, \ldots, l \rangle. \label{cycle2}
\end{align}
 The reason for the above is that $k$'s delegation  goes through $h$ to reach $j$, but then the preferred delegation from $j$ passes through $h$ (see chain $C_j$). As $k$'s delegation already includes $h$ before $j$, and an intermediator  voter cannot be repeated (definition~\ref{ref:DR}), $k$ uses another route to guru $l$ (through a less preferred option of $j$). From \eqref{cycle1} and \eqref{cycle2}, observe that there  exists a cycle in the graph, i.e. the cycle $\langle h, ..., j, ..., h \rangle$, which contradicts our assumption and proves the lemma.
 \end{proof}

\begin{theorem}
Given a voting rule $f \in \bar{F}$ and a delegation graph with no cycles,  guru participation is guaranteed to hold when using the depth-first delegation rule $d^D$.
\label{theorem:1b}
\end{theorem}

\begin{proof}
We prove this using Lemma \ref{lem:DFDLemma1} and Observation 1.
\end{proof}

We have  previously shown that depth-first delegation does not guarantee guru participation when the delegation graph contains cycles.  The next theorem states that 
breadth-first delegation always  guarantees guru participation. To show this, we  first introduce the following observation and lemma.
%The first lemma considers a situation where a voter $k$ is currently delegating her vote to a guru through the delegation chain $C_k$. Now if another voter $j$ joins the delegating electorate, and $k$ does not delegate to a guru through $j$, then $k$ delegates to the same guru as before through the same delegation chain $C_k$.

\begin{observation} \label{ob:BFD2}
Consider two voters $j$ and $k$ in a delegating electorate. Using the breadth-first delegation rule $d^B$, if $k$ is assigned guru $l$ through a delegation chain $C_k$ with $j \notin C_k$, then $k$ is assigned guru $l$ even when $j$ abstains.  This is because rule $d^B$ has used  $C_k$ ahead of any possible  delegation chain that includes $j$.
Therefore chain $C_k$ will still be used by $d^B$ when $j$ is in the abstaining electorate and no possible  delegation chain that includes $j$ can be formed.

% $j$ has no impact on which path rule $d^B$ follows to  create the chain $C_k$. Therefore the path that rule $d^B$ follows to create chain $C_k'$ is exactly the same as the path followed for $C_k$. Thus $C_k'=C_k$, which contradicts our assumption.
\end{observation}

%The next lemma shows that when voter $j$ joins the delegating electorate and is assigned guru $i$, any other voter $k$ who reassigns their guru through a delegation chain that uses $j$, must also delegate her vote to guru $i$. This is because: (1) voter $k$ cannot use a shorter chain (or same length chain with less weight) from $j$ to a guru $l \neq i$, otherwise $j$ would have $l$ as her assigned guru; and (2) if $k$ did not use the same chain from $j$ to $i$ due to shared intermediator $e$,  then $k$ would not delegate through $j$ because there is a shorter delegation chain going directly from $k$ to $e$ to the guru $i$.

%\begin{lemma}
%Consider two voters $j$ and $k$ in a delegating electorate, where $k \notin C_j= \langle j \ldots i \rangle$. Using the breadth-first delegation rule $d^B$, if $k$ delegates through $j$, then $C_k = \langle k, \ldots, j, C_{j, 2}, \ldots, C_{j, |C_j|}\rangle$.
%is assigned guru $i$ through a delegation chain $C_k$ with $j \in C_k$, then $k$ is assigned guru $l$ even when $j$ abstains.  
%\label{lem:BFD1}
%\end{lemma}

\begin{lemma}
Consider two voters $j$ and $k$ in a delegating electorate. Using the breadth-first delegation rule $d^B$, if voter $k$ is assigned her guru through a delegation chain $C_k$ with $j \in C_k$, then $k$ is assigned the same guru as $j$.  
\label{lem:BFD1}
\end{lemma}
\begin{proof}

Assume that, using $d^B$, voter $j$ is assigned guru $i$ through delegation chain $C_j = \langle j \dots,i \rangle$ and $k$ is assigned a different guru $l$ through a delegation chain that includes $j$, i.e. $C_k = \langle k, \dots, j, \dots, l \rangle$. Then either $(a)$ or $(b)$ occurs:

\begin{enumerate}[$(a)$]

\item rule $d^B$ should use chain $C'_j = \langle j, \dots,l \rangle$, which contradicts the assumption that $C_j$ is used,
\item there exists a shared intermediate voter $e$ such that 
\begin{align*}
&C_j = \langle j, \dots,e,g, \dots,i \rangle \quad \text{and} \\
&C_k = \langle k, \dots, f, j, \dots, l \rangle,
\end{align*}
 where $e \in \langle k, \dots,f\rangle$. Recall that $d^B$ prioritises shorter length delegation chains (see definition~\ref{def:BFD}). We show that voter $k$ has a  shorter delegation chain available that does not include $j$, i.e. there exists a $C_k'$ such that
 $|C_k'| < |C_k|$ and $j \notin C'_k$. Let  $C_k' =\langle k, \dots, e, g, \dots,i \rangle$.
%, thus contradicting (??iii) of Defn.~\ref{def:BFD} for $C_k = \langle k,...,f,j,...,l \rangle$. 
According to $d^B$, the delegation chain used to assign $j$'s guru, $\langle j \dots, e, g, \dots,i\rangle$, is shorter or equal in length to any other alternative, thus $|\langle j, \dots, e, g, \dots, i\rangle| \leq |\langle j, \dots, l\rangle|$. Observe that  
\begin{align*}
|\langle g, \dots, i\rangle| < |\langle j, \ldots, e, g, \dots, i\rangle| \leq |\langle j, \dots, l\rangle \Rightarrow \\
|\langle k, \dots, e\rangle| + |\langle g, \dots, i\rangle| < |\langle k, \dots, e\rangle| + |\langle j, \dots, l\rangle|. 
\end{align*}
Since $e \in \langle k, \ldots, f \rangle$, we can rewrite the previous as 
\begin{align*}
|\langle k,...,e\rangle| + |\langle g,...,i\rangle| < |\langle k,...,f\rangle| + |\langle j,...,l\rangle|.
\end{align*}
Therefore, rule $d^B$ should use $C'_k$ to assign $k$'s guru. However, since $j \notin C'_k$, the assumption is contradicted.
\end{enumerate} 
The contradictions of both $(a)$ and $(b)$ prove this lemma.
\end{proof}

%We can build on the previous two Lemma's to create our main theorem:

\begin{theorem}
Given a voting rule $f \in \bar{F}$,  guru participation is guaranteed to hold when using the breadth-first delegation rule $d^B$.
\end{theorem}
\begin{proof}
By Observation~\ref{ob:BFD2},  given voters $j$ and $k$ in the delegating electorate, if  a voter $k$ does not delegate through $j$, then $k$'s  assigned guru (if any) is the same as if $j$ abstained. By Lemma~\ref{lem:BFD1}, if a voter $k$ delegates through $j$, then the guru of $k$ is the same as the guru of $j$. Combining the above cases, we show that (regardless of $k$ delegating through $j$ or not),  whenever a voter $j$ joins the delegating electorate and is assigned to a guru $i$, then $i$ is the only casting voter who increases the number of times she becomes a guru. Since also $f \in\bar{F}$,  then Observation 1 holds, meaning that the breadth-first delegation rule $d^B$ is guaranteed to satisfy guru participation.
\end{proof}

\section{Conclusion and future work} \label{sec:con}

%Our paper, through the delegation rule function, has detailed a formalism to describe all types of ways delegations can occur. We have provided two delegation rule function instantiations,

In this paper, we discuss the depth-first and the breadth-first delegation rule  proving that only the latter  has the desirable property that every guru weakly prefers receiving delegating voting rights. However, there could be other delegation rules satisfying the same or other interesting properties that improve the concept of liquid democracy.
%There can exist many other possible delegation rules with potentially interesting properties,
% for example the Markov chain delegation rule outlined by~\citeauthor{Brill18} (\citeyear{Brill18}) to a $k$-limited depth delegation rule (where an agent's vote is only allowed to travel a $k$ depth away from it). 
Towards this path, we note that the issue of current liquid democracy implementations  suffering from a small subset of gurus representing a large part of the electorate (\citeauthor{KlingKHSS15} \citeyear{KlingKHSS15}) could be counteracted by the breadth-first delegation rule, as this rule favours keeping delegated voting rights close to their origin. We strongly believe that this hypothesis should be investigated. Other interesting future work include investigating guru participation with voting rules that do not satisfy cast participation,  relaxing the assumption of strict personal rankings over voters, and analysing other types of participation.

\bibliographystyle{ijcai19}
\bibliography{FinalBib}

\begin{thebibliography}{}

\bibitem[\protect\citeauthoryear{Behrens \bgroup \em et al.\egroup
  }{2014}]{LFBook}
Jan Behrens, Axel Kistner, Andreas Nitsche, and Bjorn Swierczek.
\newblock {\em The Principles of LiquidFeedback}.
\newblock Interaktive Demokratie, 2014.

\bibitem[\protect\citeauthoryear{Boella \bgroup \em et al.\egroup
  }{2018}]{BoellaFGKNNSSST18}
Guido Boella, Louise Francis, Elena Grassi, Axel Kistner, Andreas Nitsche,
  Alexey Noskov, Luigi Sanasi, Adriano Savoca, Claudio Schifanella, and Ioannis
  Tsampoulatidis.
\newblock Wegovnow: {A} map based platform to engage the local civic society.
\newblock In {\em Companion Proceedings of the The Web Conference, WWW 2018,
  Lyon , France, April 23-27, 2018}, pages 1215--1219, 2018.

\bibitem[\protect\citeauthoryear{Brill and Talmon}{2018}]{BrillT18}
Markus Brill and Nimrod Talmon.
\newblock Pairwise liquid democracy.
\newblock In {\em Proceedings of the Twenty-Seventh International Joint
  Conference on Artificial Intelligence, {IJCAI} 2018, July 13-19, 2018,
  Stockholm, Sweden.}, pages 137--143, 2018.

\bibitem[\protect\citeauthoryear{Brill}{2018}]{Brill18}
Markus Brill.
\newblock Interactive democracy.
\newblock In {\em Proceedings of the 17th International Conference on
  Autonomous Agents and MultiAgent Systems, {AAMAS} 2018, Stockholm, Sweden,
  July 10-15, 2018}, pages 1183--1187, 2018.

\bibitem[\protect\citeauthoryear{Christoff and Grossi}{2017}]{ChristoffG17}
Zo{\'{e}} Christoff and Davide Grossi.
\newblock Binary voting with delegable proxy: An analysis of liquid democracy.
\newblock In {\em Proceedings of the Sixteenth Conference on Theoretical
  Aspects of Rationality and Knowledge, {TARK} 2017, Liverpool, UK, 24-26 July
  2017.}, pages 134--150, 2017.

\bibitem[\protect\citeauthoryear{Fishburn and Brams}{1983}]{Fishburn}
Peter~C. Fishburn and Steven~J. Brams.
\newblock Paradoxes of preferential voting.
\newblock {\em Mathematics Magazine}, 56(4):207--214, 1983.

\bibitem[\protect\citeauthoryear{Hardt and Lopes}{2015}]{GoogleVotes}
Steve Hardt and Lia C.~R. Lopes.
\newblock Google votes: A liquid democracy experiment on a corporate social
  network.
\newblock Technical report, Technical Disclosure Commons, Google, 2015.

\bibitem[\protect\citeauthoryear{Hardt}{2014}]{GoogleVotesPresentation}
Steve Hardt.
\newblock Liquid democracy with google votes.
\newblock \url{https://www.youtube.com/watch?v=F4lkCECSBFw}, 2014.
\newblock Accessed: 2018-07-17.

\bibitem[\protect\citeauthoryear{Kahng \bgroup \em et al.\egroup
  }{2018}]{KahngMP18}
Anson Kahng, Simon Mackenzie, and Ariel~D. Procaccia.
\newblock Liquid democracy: An algorithmic perspective.
\newblock In {\em Proceedings of the Thirty-Second {AAAI} Conference on
  Artificial Intelligence, New Orleans, Louisiana, USA, February 2-7, 2018},
  2018.

\bibitem[\protect\citeauthoryear{Kling \bgroup \em et al.\egroup
  }{2015}]{KlingKHSS15}
Christoph~Carl Kling, J{\'{e}}r{\^{o}}me Kunegis, Heinrich Hartmann, Markus
  Strohmaier, and Steffen Staab.
\newblock Voting behaviour and power in online democracy: {A} study of
  liquidfeedback in germany's pirate party.
\newblock In {\em Proceedings of the Ninth International Conference on Web and
  Social Media, {ICWSM} 2015, University of Oxford, Oxford, UK, May 26-29,
  2015}, pages 208--217, 2015.

\bibitem[\protect\citeauthoryear{Moulin}{1988}]{Moulin}
Herve Moulin.
\newblock Condorcet's principle implies the no show paradox.
\newblock {\em Journal of Economic Theory}, 45:53--64, 1988.

\end{thebibliography}

\end{document}